\documentclass{article}
\usepackage[T1]{fontenc}
\usepackage{amsmath}
\usepackage{amsthm}
\usepackage{amssymb}
\usepackage{graphicx}
\usepackage{geometry}
\makeatletter

\usepackage{stmaryrd}
\usepackage{cases}

\makeatother

\providecommand\theoremname{Theorem}
\theoremstyle{plain}
\newtheorem{thm}{\protect\theoremname}
\providecommand\lemmaname{Lemma}
\newtheorem{lem}[thm]{\protect\lemmaname}

\begin{document}
\title{An attainable Gill--Massar-type bound for spin-factor models}
\author{
	Koichi Yamagata%
	\thanks{yamagata@se.kanazawa-u.ac.jp}\\
	{Institute of Science and Engineering, Kanazawa University} \\
	{Kanazawa, Ishikawa, 920-1192, Japan}
}%
\maketitle

\global\long\def\A{\mathcal{A}}%
\global\long\def\E{\mathcal{E}}%
\global\long\def\S{\mathcal{S}}%
\global\long\def\R{\mathbb{R}}%
\global\long\def\C{\mathbb{C}}%
\global\long\def\N{\mathbb{N}}%
\global\long\def\D{\mathcal{D}}%
\global\long\def\M{\mathcal{M}}%
\global\long\def\X{\mathcal{X}}%
\global\long\def\F{\mathcal{F}}%
\global\long\def\B{\mathcal{B}}%
\global\long\def\T{\mathcal{T}}%
\global\long\def\P{\mathcal{P}}%
\global\long\def\H{\mathcal{H}}%
\global\long\def\Y{\mathcal{Y}}%
\global\long\def\bra#1{\left\langle #1\right|}%
\global\long\def\ket#1{\left|#1\right\rangle }%
\global\long\def\tr{{\rm tr}\,}%
\global\long\def\Tr{{\rm Tr}\,}%
\global\long\def\braket#1#2{\left\langle #1\mid#2\right\rangle }%
\global\long\def\V{\mathcal{V}}%
\global\long\def\re{{\rm Re}\,}%
\global\long\def\im{{\rm Im}\,}%
\global\long\def\ii{\sqrt{-1}}%
\global\long\def\Span{{\rm span}}%

\begin{abstract}
We determine the exact local precision limits for single-copy estimation
of smooth multiparameter quantum statistical models contained in spin
factors, a class of matrix Jordan algebras whose state spaces generalize
the qubit Bloch ball. At any parameter point where the symmetric logarithmic
derivative (SLD) Fisher information is positive definite, we characterize
the entire attainable classical Fisher-information region over all
finite-outcome positive-operator-valued measurements. After SLD normalization,
this region consists exactly of the real symmetric positive semidefinite
matrices with trace at most one, independently of the ambient Hilbert-space
dimension. This yields a sharp weighted covariance bound for every
positive definite weight, attained by an explicit locally unbiased
estimator based on randomized spectral measurements of SLD directions.
The proof combines a statistics-preserving positive projection onto
the spin factor with the two-eigenvalue structure of its effects,
revealing the Jordan-algebraic origin of the tradeoff. The result
extends the qubit information tradeoff to models with more than three
parameters. Since the optimal measurement depends on the unknown parameter,
we simulate an adaptive scheme for a five-parameter model on a four-dimensional
Hilbert space and observe performance close to the optimal local benchmark.
\end{abstract}

\section{Introduction\label{sec:introduction}}

The simultaneous estimation of several parameters of a quantum state
requires a tradeoff between measurements that extract information
about different parameters. Determining this tradeoff, rather than
bounding the precision of each parameter separately, is a central
problem in quantum statistical inference \cite{holevo}. For qubits,
the problem admits an explicit solution in terms of the symmetric
logarithmic derivative (SLD) Fisher information. In this paper, we
establish an analogous complete solution for quantum statistical models
contained in general spin factors. These models retain the ball geometry
of the qubit state space while allowing more than three independent
parameters. We characterize their entire attainable Fisher-information
region and construct measurements attaining the optimal weighted covariance
for every positive definite weight.

Let $\mathcal{H}$ be a finite-dimensional Hilbert space, and consider
a smooth quantum statistical model $\{\rho_{\theta}\mid\theta\in\Theta\subset\mathbb{R}^{d}\}$,
where $\Theta$ is open. Throughout this paper, measurements are finite-outcome
POVMs on a single copy of the system. A POVM $M=(M_{x})_{x\in\mathcal{X}}$
induces the probabilities
\[
p_{\theta}(x;M)=\Tr(\rho_{\theta}M_{x}).
\]
An estimator is a pair $(M,\hat{\theta})$, where $\hat{\theta}:\mathcal{X}\to\mathbb{R}^{d}$
assigns an estimate to each outcome. It is locally unbiased at $\theta_{0}\in\Theta$
if
\begin{align}
\sum_{x}\hat{\theta}^{i}(x)p_{\theta_{0}}(x;M) & =\theta^{i}_{0},\nonumber \\
\sum_{x}\hat{\theta}^{i}(x)\left.\partial_{j}p_{\theta}(x;M)\right|_{\theta=\theta_{0}} & =\delta^{i}_{j},\qquad i,j=1,\ldots,d.\label{eq:local_unbiased}
\end{align}
For such an estimator, we write its covariance matrix as
\[
V_{\theta_{0}}[M,\hat{\theta}]_{ij}=\sum_{x}p_{\theta_{0}}(x;M)\bigl(\hat{\theta}^{i}(x)-\theta^{i}_{0}\bigr)\bigl(\hat{\theta}^{j}(x)-\theta^{j}_{0}\bigr).
\]
We use $\operatorname{Tr}$ for the trace on $\mathcal{H}$ and $\operatorname{tr}$
for the trace of matrices indexed by the model parameters.

The classical Fisher information obtained from $M$ is
\begin{equation}
F_{\theta_{0}}[M]_{ij}=\sum_{x:\,p_{\theta_{0}}(x;M)>0}\frac{\Tr[(\partial_{i}\rho_{\theta_{0}})M_{x}]\,\Tr[(\partial_{j}\rho_{\theta_{0}})M_{x}]}{\Tr(\rho_{\theta_{0}}M_{x})}.\label{eq:classical_fisher}
\end{equation}
Here and below, $\partial_{i}\rho_{\theta_{0}}:=\left.\partial_{i}\rho_{\theta}\right|_{\theta=\theta_{0}}$.
The restriction to positive-probability outcomes specifies the pointwise
Fisher information, including when $\rho_{\theta_{0}}$ is not faithful.
Since the model is smooth on an open parameter set, $p_{\theta_{0}}(x;M)=0$
implies $\left.\partial_{i}p_{\theta}(x;M)\right|_{\theta=\theta_{0}}=0$.
The SLDs $L_{\theta_{0},i}$ are Hermitian operators defined, together
with the SLD Fisher information matrix $J_{\theta_{0}}$, by
\begin{align}
\partial_{i}\rho_{\theta_{0}} & =\frac{1}{2}\bigl(\rho_{\theta_{0}}L_{\theta_{0},i}+L_{\theta_{0},i}\rho_{\theta_{0}}\bigr),\nonumber \\
(J_{\theta_{0}})_{ij} & =\frac{1}{2}\Tr\bigl[\rho_{\theta_{0}}(L_{\theta_{0},i}L_{\theta_{0},j}+L_{\theta_{0},j}L_{\theta_{0},i})\bigr].\label{eq:sld_fisher}
\end{align}
We assume $J_{\theta_{0}}>0$. The SLDs need not be unique for a nonfaithful
state, but $J_{\theta_{0}}$ is independent of this choice.

The classical and SLD Cram\'er--Rao inequalities give
\begin{equation}
V_{\theta_{0}}[M,\hat{\theta}]\succeq F_{\theta_{0}}[M]^{-1}\succeq J^{-1}_{\theta_{0}}\label{eq:SLD_Cramer}
\end{equation}
for every locally unbiased estimator. The first inequality is attainable
for each measurement with nonsingular Fisher information, whereas
the second is generally not attainable as a matrix equality. We therefore
minimize $\tr(WV_{\theta_{0}}[M,\hat{\theta}])$ for a real symmetric
positive definite weight matrix $W$. For qubit models, the attainable
weighted bound was established in the two-parameter case by Nagaoka
\cite{qubit_nagaoka} (see also Ref. \cite{qubit_fujiwara}) and in
the three-parameter case by Hayashi \cite{qubit_hayashi} and Gill
and Massar \cite{gill_massar}. A proof using Bessel's inequality
for the SLD inner product, together with an explicit construction
by randomized projective measurements, was given in Ref. \cite{yama_tomo}.
For a general system of dimension $D=\dim\mathcal{H}$, the Gill--Massar
inequality gives
\begin{equation}
\tr(WV_{\theta_{0}}[M,\hat{\theta}])\geq\frac{1}{D-1}\left(\tr\sqrt{J^{-1/2}_{\theta_{0}}WJ^{-1/2}_{\theta_{0}}}\right)^{2}.\label{eq:gene_gill_massar_bound}
\end{equation}
For qubits, $D=2$, this bound is attainable for every $W>0$. For
higher-dimensional models, it is not generally sharp.

We consider the corresponding estimation problem for matrix representations
of spin factors. Let $\gamma_{1},\ldots,\gamma_{q}\in\mathcal{B}_{h}(\mathcal{H})$,
$q\geq1$, be Hermitian operators satisfying
\begin{align}
\gamma_{i}\circ\gamma_{j} & :=\frac{1}{2}(\gamma_{i}\gamma_{j}+\gamma_{j}\gamma_{i})=\delta_{ij}I,\label{eq:def_spin_fac}\\
\Tr\gamma_{i} & =0,\label{eq:def_spin_fac0}
\end{align}
where $\mathcal{B}_{h}(\mathcal{H})$ denotes the real vector space
of Hermitian operators on the finite-dimensional Hilbert space $\mathcal{H}$.
The real linear space
\[
\Gamma^{q}=\operatorname{span}_{\mathbb{R}}\{I,\gamma_{1},\ldots,\gamma_{q}\}\subset\mathcal{B}_{h}(\mathcal{H})
\]
is closed under the Jordan product and is called a spin factor \cite{jordan1,jordan3}.
Its density operators form the state space
\begin{equation}
\mathcal{S}(\Gamma^{q})=\left\{ \frac{1}{D}\left(I+\sum^{q}_{i=1}\zeta^{i}\gamma_{i}\right)\,\middle|\,\boldsymbol{\zeta}\in\mathbb{R}^{q},\ \|\boldsymbol{\zeta}\|\leq1\right\} .\label{eq:spin_factor_states}
\end{equation}
Thus $q$ is the dimension of the state ball, whereas the real vector
space $\Gamma^{q}$ has dimension $q+1$. Taking $q=3$ and the Pauli
matrices recovers the qubit state space. For $q>3$, the full spin-factor
state space cannot be represented by a qubit, although it admits a
representation on a higher-dimensional quantum system.

Our main result combines the complete Fisher-information region with
the corresponding attainable covariance bound.
\begin{thm}[Fisher-information region and attainable bound]
\label{thm:main}Let $\{\rho_{\theta}\in\mathcal{S}(\Gamma^{q})\mid\theta\in\Theta\subset\mathbb{R}^{d}\}$
be a smooth model, where $\Theta$ is open and $1\leq d\leq q$. Fix
$\theta_{0}\in\Theta$ and assume $J_{\theta_{0}}>0$. The state $\rho_{\theta_{0}}$
need not be faithful. Then the following statements hold.
\begin{description}
\item [{(i)}] The attainable normalized Fisher-information region is
\begin{align}
\widehat{\mathcal{F}}_{\theta_{0}} & :=\left\{ J^{-1/2}_{\theta_{0}}F_{\theta_{0}}[M]J^{-1/2}_{\theta_{0}}\,\middle|\,M\text{ is a finite-outcome POVM on }\mathcal{H}\right\} \nonumber \\
 & =\mathcal{K}_{d}:=\left\{ K\in\mathbb{R}^{d\times d}\,\middle|\,K=K^{\mathsf{T}},\ K\succeq0,\ \operatorname{tr}K\leq1\right\} .\label{eq:main_fisher_region}
\end{align}
\item [{(ii)}] For every real symmetric positive definite $W$, every locally
unbiased estimator $(M,\hat{\theta})$ at $\theta_{0}$ satisfies
\begin{equation}
\tr(WV_{\theta_{0}}[M,\hat{\theta}])\geq c^{(SF)}_{\theta_{0},W}:=\left(\tr\sqrt{J^{-1/2}_{\theta_{0}}WJ^{-1/2}_{\theta_{0}}}\right)^{2}.\label{eq:main_spin_factor_bound}
\end{equation}
The bound is attainable. More precisely, write
\[
R:=\sqrt{J^{-1/2}_{\theta_{0}}WJ^{-1/2}_{\theta_{0}}}=U\,\operatorname{diag}(s_{1},\ldots,s_{d})\,U^{\mathsf{T}},\qquad U\in O(d).
\]
Let $M^{(i)}$ be the spectral measurement of
\begin{equation}
\hat{L}^{i}=\sum^{d}_{k=1}(U^{\mathsf{T}}J^{-1/2}_{\theta_{0}})_{ik}L_{\theta_{0},k},\qquad i=1,\ldots,d.\label{eq:optimal_sld_directions}
\end{equation}
Then the randomized measurement
\begin{equation}
M=\bigoplus^{d}_{i=1}p_{i}M^{(i)},\qquad p_{i}=\frac{s_{i}}{\sum^{d}_{j=1}s_{j}},\label{eq:optimal_randomized_measurement}
\end{equation}
admits a locally unbiased estimator attaining Eq. (\ref{eq:main_spin_factor_bound}).
\end{description}
\end{thm}

In Eq. (\ref{eq:optimal_randomized_measurement}), the measurement
label $i$ is retained as part of the outcome: if $M^{(i)}=(M^{(i)}_{x})_{x}$,
the randomized POVM has effects $(p_{i}M^{(i)}_{x})_{i,x}$. For the
resulting POVM, an attaining estimator is given, on positive-probability
outcomes, by
\begin{equation}
\hat{\theta}^{i}(x)=\theta^{i}_{0}+\sum^{d}_{k=1}\bigl(F_{\theta_{0}}[M]^{-1}\bigr)_{ik}\frac{\Tr[(\partial_{k}\rho_{\theta_{0}})M_{x}]}{\Tr(\rho_{\theta_{0}}M_{x})}.\label{eq:optimal_classic}
\end{equation}
Its values on zero-probability outcomes may be chosen arbitrarily.
Both the measurement and this estimator are designed at the fixed
point $\theta_{0}$; attainability is local. 

Theorem \ref{thm:main} identifies the complete single-copy information
tradeoff, not only its optimum for a particular weight. For $q=3$
in the Pauli representation, it recovers the qubit result. For $q>3$,
it applies to a different class of quantum statistical models on larger
Hilbert spaces. The factor $1/(D-1)$ in Eq. (\ref{eq:gene_gill_massar_bound})
is absent: once the SLD metric is accounted for, the attainable region
is independent of the dimension of the ambient representation.

The proof uses the Jordan structure of the model. The Hilbert--Schmidt
orthogonal projection onto $\Gamma^{q}$ is positive and unital, and
therefore maps every ambient POVM to a POVM with effects in $\Gamma^{q}$.
It preserves the outcome probabilities throughout the model and hence
also the classical Fisher information. This use of positive projections
connects the estimation problem with the operator-algebraic theory
of Jordan structures \cite{jordan3,effros_stormer}. Combining the
spectral structure of a spin factor with Bessel's inequality yields
$\operatorname{tr}(J^{-1}_{\theta_{0}}F_{\theta_{0}}[M])\leq1$. 

Conversely, spectral measurements of normalized SLD directions yield
normalized Fisher information matrices of the form $vv^{\mathsf{T}}$
for every unit vector $v\in\mathbb{R}^{d}$. Randomizing these measurements,
and if necessary the uninformative measurement $\{I\}$, gives all
of $\mathcal{K}_{d}$. The covariance bound and its attaining measurement
then follow from optimization over this set.

The dependence of the optimal measurement on $\theta_{0}$ motivates
an adaptive implementation \cite{adaptive_nagaoka,adaptive_fujiwara}.
Such methods have been demonstrated experimentally for photonic qubits
\cite{okamoto1,okamoto3}. We illustrate the present result numerically
for $\Gamma^{5}\subset\mathcal{B}_{h}(\mathbb{C}^{4})$. For this
five-parameter model and the weight $W=J_{\theta_{0}}$, the optimal
local cost is $5^{2}=25$, providing an explicit benchmark for the
adaptive scheme.

The remainder of the paper is organized as follows. Section \ref{sec:spin_factors}
develops the required properties of spin factors and positive projections.
Section \ref{sec:proof_main} proves Theorem \ref{thm:main}. Section
\ref{sec:simulation} presents the adaptive estimation scheme and
its numerical application to the $\Gamma^{5}$ model. Section \ref{sec:conclusion}
discusses the implications of the results and their extension to models
admitting a locally sufficient Jordan algebra of spin-factor type.
Background on the decomposition of matrix Jordan algebras is collected
in Appendix \ref{sec:real_jordan}.

\section{Spin Factors and Positive Projections\label{sec:spin_factors}}

This section develops the algebraic ingredients used in the proof
of Theorem \ref{thm:main}. We first describe the spectral structure
of spin factors and derive an operator-norm bound for their POVMs.
We then show that every POVM on the ambient Hilbert space can be replaced
by a POVM with effects in the spin factor without changing the outcome
probabilities of the model.

Throughout, $D=\dim\mathcal{H}$, and $\mathcal{B}_{h}(\mathcal{H})$
denotes the real vector space of Hermitian operators on $\mathcal{H}$.
We use $\|\cdot\|$ for the operator norm of an operator and for the
Euclidean norm of a real vector, as determined by the argument.

\subsection{Spectral structure and matrix representations\label{subsec:spin_spectral}}

Let $\gamma_{1},\ldots,\gamma_{q}$ satisfy Eqs. (\ref{eq:def_spin_fac})
and (\ref{eq:def_spin_fac0}). For $b\in\mathbb{R}^{q}$, write
\[
\gamma(b):=\sum^{q}_{j=1}b_{j}\gamma_{j}.
\]
The defining relations imply
\begin{equation}
\gamma(b)^{2}=\|b\|^{2}I,\qquad\Tr\gamma(b)=0.\label{eq:spin_clifford_identity}
\end{equation}
Thus every $A\in\Gamma^{q}$ has the form $A=aI+\gamma(b)$. If $b\ne0$,
its spectral decomposition is
\begin{equation}
A=(a+\|b\|)P_{+}(n)+(a-\|b\|)P_{-}(n),\qquad n=\frac{b}{\|b\|},\label{eq:spin_spectral_decomposition}
\end{equation}
where
\begin{equation}
P_{\pm}(n)=\frac{1}{2}\bigl(I\pm\gamma(n)\bigr),\qquad\|n\|=1.\label{eq:spin_spectral_projections}
\end{equation}
These are mutually orthogonal projections with sum $I$. Since $\Tr\gamma(n)=0$,
they satisfy
\begin{equation}
{\rm rank}P_{\pm}(n)=\Tr P_{\pm}(n)=\frac{D}{2}.\label{eq:spin_projection_rank}
\end{equation}
In particular, $D$ is even, and each of the two distinct eigenvalues
in Eq. (\ref{eq:spin_spectral_decomposition}) has multiplicity $D/2$.
If $b=0$, then $A=aI$ has only one eigenvalue. It follows that
\begin{equation}
A\geq0\quad\Longleftrightarrow\quad a\geq\|b\|.\label{eq:spin_positivity}
\end{equation}

Every projection in $\Gamma^{q}$ other than $0$ and $I$ is of the
form $P_{+}(n)$ for some unit vector $n$. Consequently, for $D>2$,
these are not rank-one projections on $\mathcal{H}$. The extreme
points of the restricted state space $\mathcal{S}(\Gamma^{q})$ are
\[
\frac{2}{D}P_{+}(n)=\frac{1}{D}\bigl(I+\gamma(n)\bigr),\qquad\|n\|=1.
\]
Their rank as density operators on $\mathcal{H}$ is also $D/2$.

For completeness, explicit matrix representations can be constructed
recursively. Let $\sigma_{1},\sigma_{2},\sigma_{3}$ be the Pauli
matrices. Given generators $\gamma_{1},\ldots,\gamma_{q}$ on $\mathcal{H}$,
define generators on $\mathbb{C}^{2}\otimes\mathcal{H}$ by
\begin{align}
\gamma'_{j} & =\sigma_{1}\otimes\gamma_{j}, &  & j=1,\ldots,q,\nonumber \\
\gamma'_{q+1} & =\sigma_{3}\otimes I, & \gamma'_{q+2} & =\sigma_{2}\otimes I.\label{eq:spin_recursive_representation}
\end{align}
They satisfy the same defining relations and generate a copy of $\Gamma^{q+2}$.
Starting with the first two Pauli matrices for $q=2$, or all three
for $q=3$, gives representations
\[
\Gamma^{q}\subset\mathcal{B}_{h}\!\left(\mathbb{C}^{2^{\lfloor q/2\rfloor}}\right),\qquad q\geq2.
\]
For $q=1$, a single Pauli matrix gives a representation on $\mathbb{C}^{2}$.
The results below apply to every representation satisfying Eqs. (\ref{eq:def_spin_fac})
and (\ref{eq:def_spin_fac0}), not only to these particular constructions.

\subsection{An operator-norm bound for spin-factor POVMs\label{subsec:spin_povm_norm}}

The two-eigenvalue structure yields a bound that is independent of
the dimension of the representation.
\begin{lem}[Operator-norm bound]
\label{lem:spin_povm_norm}For every $A\in\Gamma^{q}$ with $A\geq0$,
\begin{equation}
\frac{D}{2}\|A\|\leq\Tr A.\label{eq:tr_norm_ineq}
\end{equation}
Consequently, every POVM $N=(N_{x})_{x\in\mathcal{X}}$ with $N_{x}\in\Gamma^{q}$
satisfies
\begin{equation}
\sum_{x}\|N_{x}\|\leq2.\label{eq:POVM_ineq}
\end{equation}
\end{lem}

\begin{proof}
Write $A=aI+\gamma(b)$. By Eq. (\ref{eq:spin_positivity}), $a\geq\|b\|$,
and hence
\[
\|A\|=a+\|b\|\leq2a=\frac{2}{D}\Tr A.
\]
This proves Eq. (\ref{eq:tr_norm_ineq}). Applying it to each $N_{x}$
and using $\sum_{x}N_{x}=I$ gives
\[
\sum_{x}\|N_{x}\|\leq\frac{2}{D}\sum_{x}\Tr N_{x}=2.
\]
\end{proof}

\subsection{Positive projections and measurement reduction\label{subsec:positive_projections}}

A unital real Jordan subalgebra of $\mathcal{B}_{h}(\mathcal{H})$
is a real linear subspace $\mathcal{A}$ containing $I$ and closed
under the product
\[
A\circ B=\frac{1}{2}(AB+BA).
\]
In particular, $\Gamma^{q}$ is such an algebra: for $a,c\in\mathbb{R}$
and $b,e\in\mathbb{R}^{q}$,
\[
\bigl(aI+\gamma(b)\bigr)\circ\bigl(cI+\gamma(e)\bigr)=(ac+b\cdot e)I+\gamma(ae+cb).
\]
Equip $\mathcal{B}_{h}(\mathcal{H})$ with the real Hilbert--Schmidt
inner product
\[
\langle A,B\rangle_{\mathrm{HS}}=\Tr(AB).
\]
We use the following elementary property of the orthogonal projection
onto a unital Jordan subalgebra.
\begin{lem}[Positive orthogonal projection]
\label{lem:proj-J}Let $\mathcal{A}\subset\mathcal{B}_{h}(\mathcal{H})$
be a unital real Jordan subalgebra, and let
\[
\Pi_{\mathcal{A}}:\mathcal{B}_{h}(\mathcal{H})\longrightarrow\mathcal{A}
\]
be the Hilbert--Schmidt orthogonal projection. Then $\Pi_{\mathcal{A}}$
is positive, unital, and trace-preserving. It is also faithful: if
$B\geq0$ and $\Pi_{\mathcal{A}}(B)=0$, then $B=0$.
\end{lem}

\begin{proof}
Since $I\in\mathcal{A}$, we have $\Pi_{\mathcal{A}}(I)=I$. Orthogonality
also gives, for every $B\in\mathcal{B}_{h}(\mathcal{H})$,
\begin{equation}
\Tr\Pi_{\mathcal{A}}(B)=\langle I,\Pi_{\mathcal{A}}(B)\rangle_{\mathrm{HS}}=\langle I,B\rangle_{\mathrm{HS}}=\Tr B.\label{eq:jordan_projection_trace}
\end{equation}

To prove positivity, let $B\geq0$ and write the spectral decomposition
of $C:=\Pi_{\mathcal{A}}(B)$ as $C=\sum_{r}\lambda_{r}P_{r}$, where
the $\lambda_{r}$ are distinct. All powers of $C$ belong to $\mathcal{A}$,
because $C^{k+1}=C^{k}\circ C$. Since $I\in\mathcal{A}$, every real
polynomial in $C$ also belongs to $\mathcal{A}$. Each spectral projection
$P_{r}$ is such a polynomial, so $P_{r}\in\mathcal{A}$. Therefore
\[
\lambda_{r}\Tr P_{r}=\Tr\bigl(P_{r}\Pi_{\mathcal{A}}(B)\bigr)=\Tr(P_{r}B)\geq0.
\]
As $\Tr P_{r}>0$, all $\lambda_{r}$ are nonnegative, which proves
$\Pi_{\mathcal{A}}(B)\geq0$. Finally, if $B\geq0$ and $\Pi_{\mathcal{A}}(B)=0$,
then Eq. (\ref{eq:jordan_projection_trace}) gives $\Tr B=0$, and
hence $B=0$.
\end{proof}

For the spin factor, write $\Pi:=\Pi_{\Gamma^{q}}$. The defining
relations give $\Tr(\gamma_{i}\gamma_{j})=D\delta_{ij}$, so the projection
has the explicit form
\begin{equation}
\Pi(B)=\frac{\Tr B}{D}I+\frac{1}{D}\sum^{q}_{j=1}\Tr(B\gamma_{j})\gamma_{j},\qquad B\in\mathcal{B}_{h}(\mathcal{H}).\label{eq:spin_hs_projection}
\end{equation}
Given any POVM $M=(M_{x})_{x\in\mathcal{X}}$ on $\mathcal{H}$, set
\begin{equation}
N_{x}:=\Pi(M_{x}).\label{eq:projected_povm}
\end{equation}
By Lemma \ref{lem:proj-J}, $N_{x}\geq0$ and $\sum_{x}N_{x}=\Pi(I)=I$.
Thus $N=(N_{x})_{x}$ is a POVM with effects in $\Gamma^{q}$.

For every $\rho\in\mathcal{S}(\Gamma^{q})$, orthogonality yields
\begin{equation}
\Tr(\rho N_{x})=\langle\rho,\Pi(M_{x})\rangle_{\mathrm{HS}}=\langle\rho,M_{x}\rangle_{\mathrm{HS}}=\Tr(\rho M_{x}).\label{eq:projection_preserves_probabilities}
\end{equation}
In particular, for any smooth spin-factor model, this identity holds
at every parameter value. Differentiating gives

\begin{equation}
\Tr\bigl[(\partial_{i}\rho_{\theta})N_{x}\bigr]=\Tr\bigl[(\partial_{i}\rho_{\theta})M_{x}\bigr],\qquad i=1,\ldots,d,\label{eq:projection_preserves_derivatives}
\end{equation}
and hence
\begin{equation}
F_{\theta}[N]=F_{\theta}[M].\label{eq:projection_preserves_fisher}
\end{equation}
This remains valid when $\rho_{\theta}$ is not faithful, since both
measurements have identical probabilities and derivatives, including
the same zero-probability outcomes.

Thus restricting the measurement effects to $\Gamma^{q}$ entails
no loss of statistical information for the model. In particular, using
the same estimator function with $M$ and $N$ gives identical local
unbiasedness conditions and covariance matrices. This reduction, together
with Eq. (\ref{eq:POVM_ineq}), is the starting point for the proof
of Theorem \ref{thm:main} in Sec. \ref{sec:proof_main}.

\section{Proof of the Main Theorem\label{sec:proof_main}}

Fix $\theta_{0}\in\Theta$ and abbreviate $\rho=\rho_{\theta_{0}}$,
$L_{i}=L_{\theta_{0},i}$, and $J=J_{\theta_{0}}>0$. All measurements
and randomization probabilities constructed below are fixed at this
design point when derivatives with respect to the model parameter
are taken. We first prove the Fisher-information bound, then characterize
its entire attainable region, and finally optimize the weighted covariance
over this region.

\subsection{The Fisher-information bound\label{subsec:fisher_bound}}

On $\mathcal{B}_{h}(\mathcal{H})$, define the real symmetric bilinear
form
\begin{equation}
\langle A,B\rangle_{\rho}:=\frac{1}{2}\Tr\bigl[\rho(AB+BA)\bigr].\label{eq:sld_bilinear_form}
\end{equation}
This form is positive semidefinite, since $\langle A,A\rangle_{\rho}=\Tr(\rho A^{2})\geq0$.
It need not be an inner product when $\rho$ is not faithful. The
SLD equations imply
\[
\langle I,L_{i}\rangle_{\rho}=0,\qquad\langle L_{i},L_{j}\rangle_{\rho}=J_{ij}.
\]
Introduce normalized SLDs
\begin{equation}
\widetilde{L}_{i}:=\sum^{d}_{j=1}(J^{-1/2})_{ij}L_{j},\qquad i=1,\ldots,d,\qquad\widetilde{L}_{0}:=I.\label{eq:normalized_slds}
\end{equation}
They satisfy
\begin{equation}
\langle\widetilde{L}_{i},\widetilde{L}_{j}\rangle_{\rho}=\delta_{ij},\qquad i,j=0,\ldots,d.\label{eq:normalized_sld_orthogonality}
\end{equation}
Consequently, for every $A\in\mathcal{B}_{h}(\mathcal{H})$,
\begin{equation}
\sum^{d}_{i=0}\langle\widetilde{L}_{i},A\rangle^{2}_{\rho}\leq\langle A,A\rangle_{\rho}.\label{eq:sld_bessel}
\end{equation}
Indeed, this follows by expanding the nonnegative quantity
\[
\left\langle A-\sum^{d}_{i=0}\langle\widetilde{L}_{i},A\rangle_{\rho}\widetilde{L}_{i},A-\sum^{d}_{i=0}\langle\widetilde{L}_{i},A\rangle_{\rho}\widetilde{L}_{i}\right\rangle _{\rho}.
\]
Thus Bessel's inequality remains valid here without assuming that
$\rho$ is faithful.

For a POVM $M=(M_{x})_{x}$, write $p_{x}:=\Tr(\rho M_{x})$ and define
\begin{equation}
\widehat{F}[M]:=J^{-1/2}F_{\theta_{0}}[M]J^{-1/2}.\label{eq:hat_F1}
\end{equation}
Using the SLD equations in Eq. (\ref{eq:classical_fisher}) gives
\begin{equation}
\widehat{F}[M]_{ij}=\sum_{x:\,p_{x}>0}\frac{\langle\widetilde{L}_{i},M_{x}\rangle_{\rho}\langle\widetilde{L}_{j},M_{x}\rangle_{\rho}}{p_{x}},\qquad i,j=1,\ldots,d.\label{eq:normalized_fisher_entries}
\end{equation}
In particular, $\widehat{F}[M]\succeq0$.
\begin{thm}[Spin-factor Fisher-information bound]
\label{thm:GM_ineq}Every POVM $M$ on $\mathcal{H}$ satisfies
\begin{equation}
\tr\widehat{F}[M]=\tr\bigl(J^{-1}F_{\theta_{0}}[M]\bigr)\leq1.\label{eq:spin_fisher_trace_bound}
\end{equation}
\end{thm}

\begin{proof}
By the measurement reduction in Sec. \ref{subsec:positive_projections},
the POVM $N=(\Pi(M_{x}))_{x}$ has effects in $\Gamma^{q}$ and satisfies
$F_{\theta_{0}}[N]=F_{\theta_{0}}[M]$. It therefore suffices to prove
the bound for $N$. Set $p_{x}=\Tr(\rho N_{x})$ and $\mathcal{X}_{+}=\{x\mid p_{x}>0\}$.
Since $\langle I,N_{x}\rangle_{\rho}=p_{x}$, Eq. (\ref{eq:sld_bessel})
implies
\[
\sum^{d}_{i=1}\langle\widetilde{L}_{i},N_{x}\rangle^{2}_{\rho}\leq\Tr(\rho N^{2}_{x})-p^{2}_{x}.
\]
Moreover, positivity of $N_{x}$ gives $N^{2}_{x}\preceq\|N_{x}\|N_{x}$.
Using $\sum_{x\in\mathcal{X}_{+}}p_{x}=1$, we obtain
\begin{align}
\tr\widehat{F}[N] & =\sum_{x\in\mathcal{X}_{+}}\frac{\sum^{d}_{i=1}\langle\widetilde{L}_{i},N_{x}\rangle^{2}_{\rho}}{p_{x}}\nonumber \\
 & \leq\sum_{x\in\mathcal{X}_{+}}\frac{\Tr(\rho N^{2}_{x})}{p_{x}}-1\nonumber \\
 & \leq\sum_{x\in\mathcal{X}_{+}}\|N_{x}\|-1\nonumber \\
 & \leq\sum_{x}\|N_{x}\|-1\leq1,\label{eq:spin_fisher_bound_proof}
\end{align}
where the last inequality is Eq. (\ref{eq:POVM_ineq}). Only positive-probability
outcomes appear in the denominators, so the argument also applies
to nonfaithful states.
\end{proof}

\subsection{The attainable Fisher-information region\label{subsec:attainable_fisher_region}}

We next show that the trace constraint in Theorem \ref{thm:GM_ineq}
completely characterizes the attainable normalized Fisher information.
\begin{lem}[Spectral measurements of normalized SLD directions]
\label{lem:one_para}For a unit vector $v\in\mathbb{R}^{d}$, let
$M^{(v)}$ be the spectral measurement of
\begin{equation}
L(v):=\sum^{d}_{i=1}v_{i}\widetilde{L}_{i}.\label{eq:sld_direction}
\end{equation}
Then
\begin{equation}
\widehat{F}[M^{(v)}]=vv^{\mathsf{T}}.\label{eq:rank_one_fisher}
\end{equation}
\end{lem}

\begin{proof}
Write $L(v)=\sum_{x}\lambda_{x}P_{x}$ and $M^{(v)}=(P_{x})_{x}$,
and set $p_{x}=\Tr(\rho P_{x})$. For each $x$, the spectral relation
$L(v)P_{x}=P_{x}L(v)=\lambda_{x}P_{x}$ gives
\[
\langle L(v),P_{x}\rangle_{\rho}=\lambda_{x}p_{x}.
\]
Therefore, by Eq. (\ref{eq:normalized_fisher_entries}),
\begin{align}
v^{\mathsf{T}}\widehat{F}[M^{(v)}]v & =\sum_{x:\,p_{x}>0}\frac{\langle L(v),P_{x}\rangle^{2}_{\rho}}{p_{x}}\nonumber \\
 & =\sum_{x}\lambda^{2}_{x}p_{x}=\Tr\bigl(\rho L(v)^{2}\bigr)=1.\label{eq:directional_fisher_attainment}
\end{align}
The zero-probability outcomes contribute zero to the second sum, and
the last equality follows from Eq. (\ref{eq:normalized_sld_orthogonality})
and $\|v\|=1$. On the other hand, $\widehat{F}[M^{(v)}]\succeq0$
and Theorem \ref{thm:GM_ineq} gives $\tr\widehat{F}[M^{(v)}]\leq1$.
In an orthonormal basis beginning with $v$, all other diagonal entries
must therefore vanish. Positivity then forces all rows and columns
corresponding to $v^{\perp}$ to vanish, proving Eq. (\ref{eq:rank_one_fisher}).
\end{proof}

The rank-one property in Eq. (\ref{eq:rank_one_fisher}) concerns
the $d\times d$ normalized Fisher information matrix, not the spectral
projections on $\mathcal{H}$. The proof requires only the SLD equations
and remains valid for any choice of Hermitian SLD representatives
at a nonfaithful state.
\begin{lem}[Complete Fisher-information region]
\label{lem:fisher_set} The attainable set $\widehat{\mathcal{F}}_{\theta_{0}}$
defined in Eq. (\ref{eq:main_fisher_region}) is convex and equals
$\mathcal{K}_{d}$.
\end{lem}

\begin{proof}
Let $M$ and $N$ be two POVMs and $t\in[0,1]$. Their randomized
measurement $tM\oplus(1-t)N$ has effects $(tM_{x})_{x}$ and $((1-t)N_{y})_{y}$,
with the measurement label retained in the outcome. Since $t$ is
fixed when differentiating, the Fisher information satisfies
\begin{equation}
\widehat{F}[tM\oplus(1-t)N]=t\widehat{F}[M]+(1-t)\widehat{F}[N].\label{eq:fisher_randomization}
\end{equation}
This also holds at $t=0$ and $t=1$, with the zero-weight effects
omitted. Thus the attainable set is convex. The inclusion $\widehat{\mathcal{F}}_{\theta_{0}}\subset\mathcal{K}_{d}$
follows from Theorem \ref{thm:GM_ineq}.

Conversely, let $K\in\mathcal{K}_{d}$ and write
\[
K=\sum^{d}_{a=1}\mu_{a}v_{a}v^{\mathsf{T}}_{a},\qquad\mu_{a}\geq0,\qquad\sum^{d}_{a=1}\mu_{a}\leq1,
\]
where $v_{1},\ldots,v_{d}$ are orthonormal. By Lemma \ref{lem:one_para},
each $v_{a}v^{\mathsf{T}}_{a}$ is attained by $M^{(v_{a})}$. The
one-outcome measurement $\{I\}$ has zero Fisher information. Hence
the finite-outcome POVM
\[
\left(\bigoplus^{d}_{a=1}\mu_{a}M^{(v_{a})}\right)\oplus\left(1-\sum^{d}_{a=1}\mu_{a}\right)\{I\}
\]
has normalized Fisher information $K$ by Eq. (\ref{eq:fisher_randomization}).
This proves the reverse inclusion and establishes part (i) of Theorem
\ref{thm:main}.
\end{proof}

\subsection{\label{subsec:weighted_optimization}Weighted optimization}

The remaining optimization is a matrix problem on $\mathcal{K}_{d}$.
We explicitly restrict to positive definite $K$ when taking its inverse.
\begin{lem}[Optimal weighted inverse]
\label{lem:opt} For every real symmetric positive definite matrix
$G$,
\begin{equation}
\min_{\substack{K\in\mathcal{K}_{d}\\
K>0
}
}\tr(GK^{-1})=\bigl(\tr\sqrt{G}\bigr)^{2}.\label{eq:weighted_matrix_minimum}
\end{equation}
The unique minimizer is
\begin{equation}
K_{*}:=\frac{\sqrt{G}}{\tr\sqrt{G}}.\label{eq:opt_k}
\end{equation}
\end{lem}

\begin{proof}
Apply the Hilbert--Schmidt Cauchy--Schwarz inequality to $K^{1/2}$
and $K^{-1/2}G^{1/2}$. It gives
\begin{equation}
\bigl(\tr\sqrt{G}\bigr)^{2}\leq\tr K\,\tr(GK^{-1})\leq\tr(GK^{-1}).\label{eq:weighted_matrix_schwarz}
\end{equation}
Equality in the first inequality holds exactly when $K^{-1/2}G^{1/2}=cK^{1/2}$
for some $c>0$, or equivalently $G^{1/2}=cK$. Equality in the second
requires $\tr K=1$. Together these conditions give Eq. (\ref{eq:opt_k}),
which is positive definite and belongs to $\mathcal{K}_{d}$.
\end{proof}

\subsection{\label{subsec:main_attainability}Completion of the proof and attainability}

\begin{proof}[Proof of Theorem \ref{thm:main}]
 Part (i) has been proved in Lemma \ref{lem:fisher_set}. For part
(ii), we first justify the use of inverse Fisher information for a
locally unbiased estimator. For an arbitrary POVM $M$, let $p_{x}=p_{\theta_{0}}(x;M)$
and define the score vector on positive-probability outcomes by
\begin{equation}
\ell(x)_{i}:=\frac{\Tr[(\partial_{i}\rho_{\theta_{0}})M_{x}]}{p_{x}},\qquad p_{x}>0,\qquad i=1,\ldots,d.\label{eq:score_vector}
\end{equation}
Set $\ell(x)=0$ on zero-probability outcomes and write $\mathbb{E}_{\theta_{0}}$
for expectation under $p_{\theta_{0}}(\cdot;M)$. Because $\Theta$
is open and the model is smooth, $p_{x}=0$ implies $\partial_{i}p_{\theta_{0}}(x;M)=0$
for every $i$. Consequently,
\begin{equation}
\mathbb{E}_{\theta_{0}}[\ell]=0,\qquad\mathbb{E}_{\theta_{0}}[\ell\ell^{\mathsf{T}}]=F_{\theta_{0}}[M].\label{eq:score_moments}
\end{equation}

Suppose $(M,\hat{\theta})$ is locally unbiased at $\theta_{0}$,
and set $z(x)=\hat{\theta}(x)-\theta_{0}$. The local unbiasedness
conditions imply
\begin{equation}
\mathbb{E}_{\theta_{0}}[z]=0,\qquad\mathbb{E}_{\theta_{0}}[z\ell^{\mathsf{T}}]=I_{d}.\label{eq:unbiased_score_identity}
\end{equation}
Write $F=F_{\theta_{0}}[M]$ and $V=V_{\theta_{0}}[M,\hat{\theta}]$.
If $a\in\ker F$, then $\ell(x)^{\mathsf{T}}a=0$ whenever $p_{x}>0$.
Equation (\ref{eq:unbiased_score_identity}) therefore gives $a=0$.
Thus $F>0$, and
\begin{equation}
V-F^{-1}=\mathbb{E}_{\theta_{0}}\bigl[(z-F^{-1}\ell)(z-F^{-1}\ell)^{\mathsf{T}}\bigr]\succeq0.\label{eq:classical_covariance_residual}
\end{equation}

Now set
\[
G:=J^{-1/2}WJ^{-1/2},\qquad K:=\widehat{F}[M].
\]
By part (i), $K\in\mathcal{K}_{d}$, and $F>0$ implies $K>0$. Equation
(\ref{eq:classical_covariance_residual}) and Lemma \ref{lem:opt}
yield
\begin{align}
\tr(WV) & \geq\tr(WF^{-1})=\tr(GK^{-1})\nonumber \\
 & \geq\bigl(\tr\sqrt{G}\bigr)^{2}=c^{(SF)}_{\theta_{0},W}.\label{eq:main_bound_proved}
\end{align}

To attain this bound, write
\[
R:=\sqrt{G}=U\,{\rm diag}(s_{1},\ldots,s_{d})\,U^{\mathsf{T}},\qquad U\in O(d),\qquad s_{i}>0,
\]
and let $u_{i}$ be the $i$th column of $U$. The corresponding normalized
SLD direction is
\[
L(u_{i})=\sum^{d}_{a=1}U_{ai}\widetilde{L}_{a}=\sum^{d}_{k=1}(U^{\mathsf{T}}J^{-1/2})_{ik}L_{k}=\hat{L}^{i},
\]
which is precisely the operator in Eq. (\ref{eq:optimal_sld_directions}).
Its spectral measurement $M^{(i)}=M^{(u_{i})}$ satisfies $\widehat{F}[M^{(i)}]=u_{i}u^{\mathsf{T}}_{i}$
by Lemma \ref{lem:one_para}. Choose the randomization probabilities
$p_{i}=s_{i}/\sum^{d}_{j=1}s_{j}$ as in Eq. (\ref{eq:optimal_randomized_measurement}).
The resulting POVM $M_{*}=\bigoplus^{d}_{i=1}p_{i}M^{(i)}$, with
the measurement label recorded, has
\begin{equation}
\widehat{F}[M_{*}]=\sum^{d}_{i=1}p_{i}u_{i}u^{\mathsf{T}}_{i}=\frac{R}{\tr R}=K_{*}.\label{eq:optimalF}
\end{equation}
In particular, $F_{*}:=F_{\theta_{0}}[M_{*}]>0$.

Finally, let $\ell_{*}(x)$ be the score of $M_{*}$ at $\theta_{0}$
and choose
\begin{equation}
\hat{\theta}_{*}(x)=\theta_{0}+F^{-1}_{*}\ell_{*}(x)\quad\text{if }p_{\theta_{0}}(x;M_{*})>0.\label{eq:attaining_score_estimator}
\end{equation}
Assign arbitrary finite values on zero-probability outcomes. Equation
(\ref{eq:score_moments}) shows that this estimator is locally unbiased
and has covariance $F^{-1}_{*}$. It is the estimator stated in Eq.
(\ref{eq:optimal_classic}). Since $\widehat{F}[M_{*}]=K_{*}$, equality
holds in Eq. (\ref{eq:main_bound_proved}). This proves part (ii)
and completes the proof.
\end{proof}

\section{\label{sec:simulation}Adaptive Estimation for the $\Gamma^{5}$
Model}

We illustrate the attainable bound with a numerical study of a five-parameter
spin-factor model on $\mathbb{C}^{4}$. The optimal measurement depends
on the unknown state, so we use a blockwise adaptive scheme: the measurement
is selected using the current estimate and updated after each block
of observations. The purpose is to compare its finite-sample performance
with the local benchmark established in Theorem \ref{thm:main}.

\subsection{\label{subsec:gamma5_model}The model and its SLD Fisher information}

Consider
\begin{equation}
\rho_{\theta}=\frac{1}{4}\bigl(I_{4}+\gamma(\theta)\bigr),\qquad\gamma(\theta)=\sum^{5}_{j=1}\theta^{j}\gamma_{j},\qquad\theta\in\Theta:=\{\theta\in\mathbb{R}^{5}\mid\|\theta\|<1\},\label{eq:gamma5_model}
\end{equation}
where
\begin{align}
\gamma_{1} & =\sigma_{1}\otimes\sigma_{1}, & \gamma_{2} & =\sigma_{1}\otimes\sigma_{2}, & \gamma_{3} & =\sigma_{1}\otimes\sigma_{3},\nonumber \\
\gamma_{4} & =\sigma_{3}\otimes I_{2}, & \gamma_{5} & =\sigma_{2}\otimes I_{2}.\label{eq:gamma5_generators}
\end{align}
Here $I_{k}$ denotes the identity matrix of size $k$, and the $\sigma_{j}$
are the Pauli matrices. The true parameter used in the simulations
is
\begin{equation}
\theta_{0}=\frac{1}{10}(1,2,3,4,5)^{\mathsf{T}},\qquad\|\theta_{0}\|^{2}=0.55.\label{eq:simulation_true_parameter}
\end{equation}
All states in $\Theta$ are faithful. We also use the continuous extension
of $\rho_{\theta}$ to the closed ball $\overline{\Theta}$ when performing
constrained maximum-likelihood estimation.

For $\theta\in\Theta$, the SLDs and their Fisher information are
\begin{align}
L_{\theta,i} & =\gamma_{i}+\frac{\theta^{i}}{1-\|\theta\|^{2}}\bigl(\gamma(\theta)-I_{4}\bigr),\label{eq:gamma5_slds}\\
J_{\theta} & =I_{5}+\frac{\theta\theta^{\mathsf{T}}}{1-\|\theta\|^{2}},\qquad J^{-1}_{\theta}=I_{5}-\theta\theta^{\mathsf{T}}.\label{eq:gamma5_sld_information}
\end{align}
These expressions follow by substitution into the SLD equations, using
the defining relations of the generators. Choosing the weight $W_{\theta}=J_{\theta}$
in Theorem \ref{thm:main} gives
\begin{equation}
c^{(SF)}_{\theta,J_{\theta}}=\left(\tr\sqrt{J^{-1/2}_{\theta}J_{\theta}J^{-1/2}_{\theta}}\right)^{2}=25.\label{eq:gamma5_cost}
\end{equation}

\subsection{\label{subsec:gamma5_measurement}An explicit optimal measurement}

For this weight, the randomization probabilities in Eq. (\ref{eq:optimal_randomized_measurement})
are all $1/5$. An explicit realization uses one radial direction
and four tangential directions. Given a nonzero design point $\vartheta\in\overline{\Theta}$,
choose an orthonormal basis $n_{1}(\vartheta),\ldots,n_{5}(\vartheta)$
of $\mathbb{R}^{5}$ with
\[
n_{1}(\vartheta)=\frac{\vartheta}{\|\vartheta\|}.
\]
At $\vartheta=0$, any orthonormal basis can be used. For each $i$,
define the binary projective measurement
\begin{equation}
P_{i,\pm}(\vartheta)=\frac{1}{2}\bigl(I_{4}\pm\gamma(n_{i}(\vartheta))\bigr).\label{eq:gamma5_binary_measurement}
\end{equation}
Each of these projections has rank two. Selecting $i$ uniformly and
retaining both $i$ and the binary outcome gives the ten-outcome POVM
\begin{equation}
M_{i,\pm}(\vartheta)=\frac{1}{5}P_{i,\pm}(\vartheta),\qquad i=1,\ldots,5.\label{eq:gamma5_randomized_povm}
\end{equation}
The POVM effects $M_{i,\pm}$ are scaled projections, not projections
themselves. For a state with parameter $\theta$ and a fixed design
$\vartheta$, the outcome probabilities are
\begin{equation}
p_{\theta}(i,\pm\mid\vartheta)=\Tr\bigl[\rho_{\theta}M_{i,\pm}(\vartheta)\bigr]=\frac{1}{10}\bigl(1\pm\theta^{\mathsf{T}}n_{i}(\vartheta)\bigr).\label{eq:gamma5_outcome_probabilities}
\end{equation}
For interior $\theta$, direct differentiation yields
\begin{equation}
F_{\theta}[M(\vartheta)]=\frac{1}{5}\sum^{5}_{i=1}\frac{n_{i}(\vartheta)n_{i}(\vartheta)^{\mathsf{T}}}{1-\bigl(\theta^{\mathsf{T}}n_{i}(\vartheta)\bigr)^{2}}.\label{eq:gamma5_design_information}
\end{equation}
At $\vartheta=\theta\ne0$, the radial term has denominator $1-\|\theta\|^{2}$
and the four tangential terms have denominator one. Therefore, including
$\theta=0$,
\begin{equation}
F_{\theta}[M(\theta)]=\frac{1}{5}J_{\theta},\qquad\tr\bigl(J_{\theta}F_{\theta}[M(\theta)]^{-1}\bigr)=25.\label{eq:gamma5_measurement_optimality}
\end{equation}
For $n$ repetitions of this fixed, locally optimal measurement, the
inverse Fisher information is $5J^{-1}_{\theta}/n$. Thus $25$ is
the corresponding sample-size-scaled local benchmark.

The measurement in Eq. (\ref{eq:gamma5_randomized_povm}) is also
well-defined when the design point lies on the boundary of the ball:
it uses only an orthonormal frame and does not require evaluating
$J_{\vartheta}$ there. The optimality statement in Eq. (\ref{eq:gamma5_measurement_optimality}),
however, is made only for interior parameters.

\subsection{\label{subsec:blockwise_adaptive}Blockwise adaptive maximum-likelihood
estimation}

The optimal measurement $M(\theta_{0})$ obtained from Theorem \ref{thm:main}
depends on the unknown true parameter $\theta_{0}$. To implement
this construction without prior knowledge of $\theta_{0}$, we adopt
an adaptive estimation scheme \cite{adaptive_nagaoka,adaptive_fujiwara}.
In this scheme, the true parameter is replaced by the current maximum-likelihood
estimate, and the corresponding measurement is used for subsequent
observations. Fujiwara \cite{adaptive_fujiwara} established the strong
consistency and asymptotic efficiency of adaptive maximum-likelihood
estimation under suitable regularity conditions.

In our simulations, the measurement is updated after each block of
$b=500$ observations, rather than after every observation, to reduce
the computational cost. The initial design is the canonical frame,
corresponding to the seed $\hat{\theta}^{(0)}=(1,0,0,0,0)^{\mathsf{T}}$.
After $kb$ observations, the estimate $\hat{\theta}^{(kb)}$ is used
to select the measurement for the next block,
\[
kb+1,\ldots,(k+1)b.
\]
Thus the outcome at time $t$ is obtained with design
\begin{equation}
\vartheta_{t}=\hat{\theta}^{(b\lfloor(t-1)/b\rfloor)}.\label{eq:adaptive_design_index}
\end{equation}
Writing the observed outcome as $X_{t}=(i_{t},\varepsilon_{t})$,
with $\varepsilon_{t}\in\{-1,+1\}$, the cumulative log likelihood
is
\begin{equation}
\ell_{n}(\theta)=\sum^{n}_{t=1}\log\left[\frac{1+\varepsilon_{t}\theta^{\mathsf{T}}n_{i_{t}}(\vartheta_{t})}{10}\right].\label{eq:adaptive_log_likelihood}
\end{equation}
We use the convention $\log0=-\infty$ at boundary candidate states.
All realized design settings are held fixed in this likelihood; they
are functions of the observed past, not of the candidate parameter
being optimized. At the end of each block, we compute a constrained
maximum-likelihood estimate
\begin{equation}
\hat{\theta}^{(kb)}\in\operatorname*{arg\,max}_{\theta\in\overline{\Theta}}\ell_{kb}(\theta),\qquad k=1,2,\ldots.\label{eq:adaptive_mle}
\end{equation}
The maximization uses all preceding observations, not just the most
recent block. The closed-ball constraint allows boundary maximizers,
while the true parameter in Eq. (\ref{eq:simulation_true_parameter})
is strictly interior. 

Each run consists of ten blocks, for a total of $n=5000$ copies.
Figure \ref{fig:each} shows the five parameter estimates for ten
runs, recorded at $n=500,1000,\ldots,5000$. The trajectories become
concentrated near the true parameter values as the sample size increases.

\begin{figure}
\centering{}\includegraphics[scale=0.5]{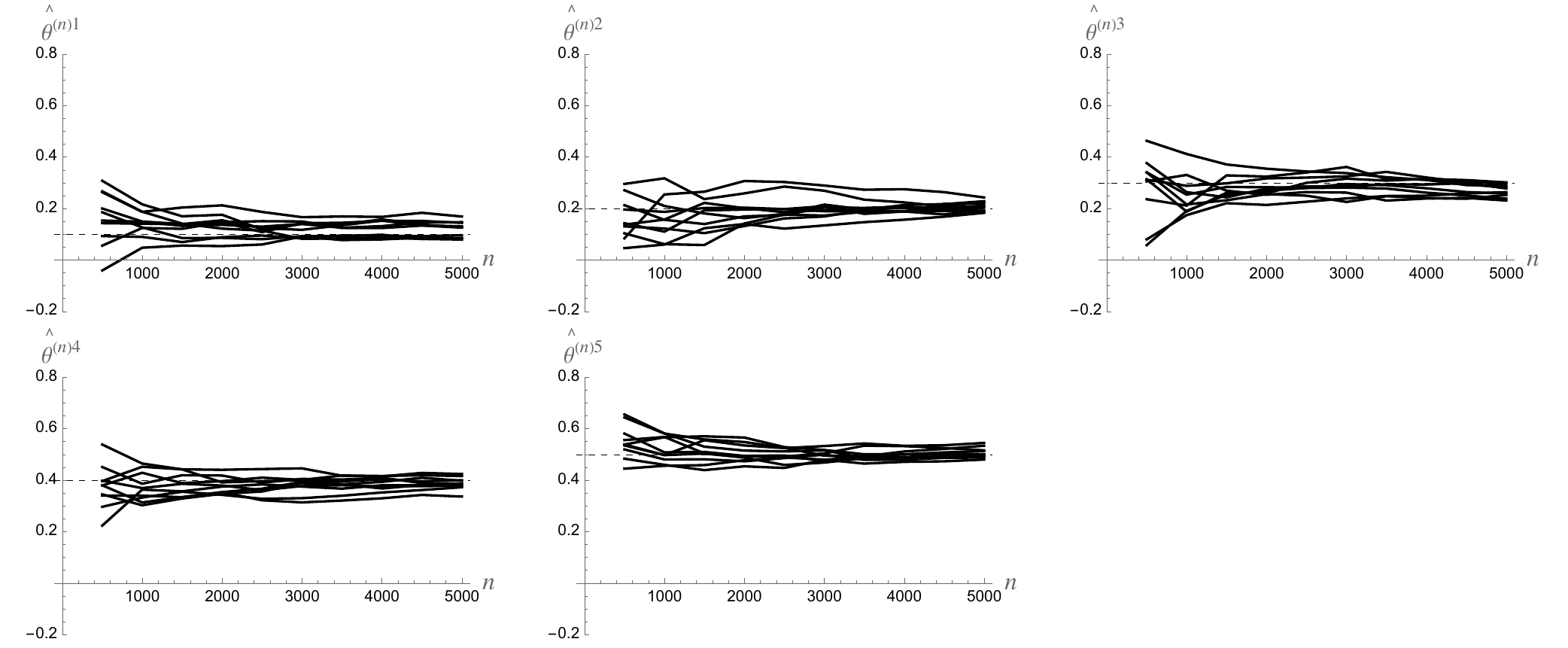}
\caption{Blockwise adaptive estimation for the $\Gamma^{5}$ model. Each panel
shows one coordinate of $\hat{\theta}^{(n)}$ for ten runs with $\theta_{0}=(1,2,3,4,5)^{\mathsf{T}}/10$.
Estimates and measurement settings are updated after every $500$
observations, up to $n=5000$. Dashed horizontal lines indicate the
corresponding true values. Lines between recorded estimates are guides
to the eye.\label{fig:each}}
\end{figure}

\subsection{\label{subsec:bures_performance}Bures loss and numerical performance}

To quantify the estimation error, we use the squared Bures distance
as the loss
\begin{equation}
B(\theta,\theta'):=2\left(1-\Tr\sqrt{\rho^{1/2}_{\theta}\rho_{\theta'}\rho^{1/2}_{\theta}}\right).\label{eq:bures_loss}
\end{equation}
For the present spin-factor model, the root fidelity has the explicit
form
\begin{equation}
\Tr\sqrt{\rho^{1/2}_{\theta}\rho_{\theta'}\rho^{1/2}_{\theta}}=\sqrt{\frac{1+\theta^{\mathsf{T}}\theta'+\sqrt{(1-\|\theta\|^{2})(1-\|\theta'\|^{2})}}{2}}.\label{eq:gamma5_root_fidelity}
\end{equation}
This follows from the two-eigenvalue structure of the spin factor
and extends continuously to $\overline{\Theta}$. Expanding about
an interior parameter gives
\begin{equation}
4B(\theta,\theta+\delta)=\delta^{\mathsf{T}}J_{\theta}\delta+O(\|\delta\|^{3}).\label{eq:bures_local_expansion}
\end{equation}
Thus, in the local regime, $4B$ measures the squared estimation error
weighted by the SLD Fisher information.

The left panel of Fig. \ref{fig:bures} shows $B(\theta_{0},\hat{\theta}^{(n)})$
for the same ten runs as Fig. \ref{fig:each}. To estimate its mean,
a separate set of $N_{\mathrm{MC}}=1000$ runs was performed with
the same true parameter and block size. At each recorded sample size,
define
\begin{equation}
\overline{B}_{n}:=\frac{1}{N_{\mathrm{MC}}}\sum^{N_{\mathrm{MC}}}_{r=1}B\bigl(\theta_{0},\hat{\theta}^{(n)}_{r}\bigr),\label{eq:mean_bures_loss}
\end{equation}
where $r$ labels the run. The right panel of Fig. \ref{fig:bures}
plots $4n\overline{B}_{n}$ together with the benchmark $25$.

The distinction between mean squared error and covariance is important
at finite sample sizes. If $\delta_{n}=\hat{\theta}^{(n)}-\theta_{0}$,
$b_{n}=\mathbb{E}_{\theta_{0}}[\hat{\theta}^{(n)}]-\theta_{0}$, and
$C_{n}=\operatorname{Cov}_{\theta_{0}}(\hat{\theta}^{(n)})$, then
\begin{equation}
\mathbb{E}_{\theta_{0}}[\delta^{\mathsf{T}}_{n}J_{\theta_{0}}\delta_{n}]=\operatorname{tr}(J_{\theta_{0}}C_{n})+b^{\mathsf{T}}_{n}J_{\theta_{0}}b_{n}.\label{eq:weighted_mse_decomposition}
\end{equation}
Accordingly, $4n\overline{B}_{n}$ is a Monte Carlo estimate of the
scaled mean Bures loss, not a direct estimate of the scaled covariance.
Its comparison with $25$ is motivated by Eq. (\ref{eq:bures_local_expansion})
and the inverse-information benchmark in Eq. (\ref{eq:gamma5_measurement_optimality}).

At larger sample sizes, the scaled mean loss $4n\overline{B}_{n}$
remains close to the theoretical benchmark of $25$. These results
illustrate performance close to the local benchmark in a five-parameter
model beyond the qubit case.

\begin{figure}
\begin{centering}
\includegraphics[scale=0.5]{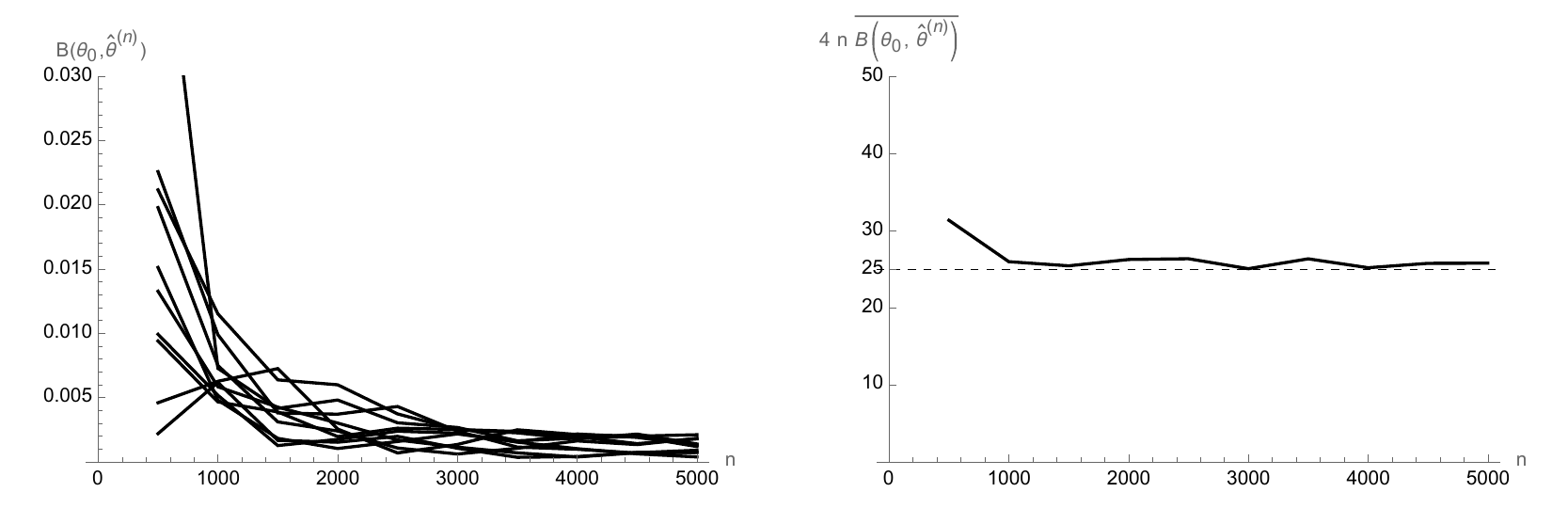}
\par\end{centering}
\caption{Bures-loss performance of the blockwise adaptive scheme. Left: the
squared Bures loss $B(\theta_{0},\hat{\theta}^{(n)})$, defined in
Eq. (\ref{eq:bures_loss}), for the ten runs in Fig. \ref{fig:each}.
Right: $4n\overline{B}_{n}$, evaluated using a separate set of $1000$
runs. The dashed line marks the local benchmark $25$. Both panels
use block size $500$ and total sample size $5000$.\label{fig:bures}}
\end{figure}

\section{\label{sec:conclusion}Discussion and Conclusion}

We have characterized the complete single-copy Fisher-information
region for smooth quantum statistical models contained in spin factors.
After normalization by the SLD Fisher information, the attainable
region is exactly $\mathcal{K}_{d}$, the set of real symmetric positive
semidefinite matrices with trace at most one. This characterization
yields the sharp weighted covariance bound in Theorem \ref{thm:main}
for every positive definite weight, together with an explicit attaining
measurement obtained by randomizing spectral measurements of normalized
SLD directions. The result applies to all finite-outcome POVMs on
the ambient Hilbert space and does not require the state at the estimation
point to be faithful, provided that the SLD Fisher information is
positive definite.

The information tradeoff is determined by the spin-factor structure,
rather than by the dimension of its matrix representation. The positive
projection onto the spin factor reduces arbitrary measurements to
spin-factor-valued POVMs without changing their outcome probabilities.
The spectral structure of these effects then gives the dimension-independent
constraint $\tr(J^{-1}_{\theta_{0}}F_{\theta_{0}}[M])\leq1$. For
the Pauli representation of $\Gamma^{3}$, this reproduces the qubit
tradeoff; for $q>3$, it also covers models with more than three independent
parameters on higher-dimensional Hilbert spaces. In particular, the
ambient-dimension factor appearing in the general Gill--Massar bound
is absent from the sharp spin-factor bound. For the SLD weight $W=J_{\theta_{0}}$,
the optimal Fisher information is $J_{\theta_{0}}/d$ and the minimum
weighted covariance cost is $d^{2}$.

The optimal measurement depends on the unknown true parameter, which
naturally leads to an adaptive implementation in which the measurement
is selected using the current estimate. Fujiwara's results on strong
consistency and asymptotic efficiency under suitable regularity conditions
provide the theoretical foundation for adaptive quantum estimation
\cite{adaptive_fujiwara}. Section \ref{sec:simulation} applies this
approach to the five-parameter $\Gamma^{5}$ model. At larger sample
sizes, the scaled mean Bures loss $4n\overline{B}_{n}$ remains close
to the theoretical benchmark of $25$, illustrating the use of the
optimal local measurement in estimation beyond the qubit setting.

The proof also gives an extension to models that admit a locally sufficient
spin-factor algebra, even when their density operators do not themselves
belong to that algebra. In the framework of Jordan-algebraic sufficiency
\cite{yamagata-jordan}, the relevant condition can be expressed as
follows. Consider a smooth quantum statistical model with $J_{\theta_{0}}>0$,
and suppose that there is a positive unital real-linear map
\[
\alpha:\mathcal{B}_{h}(\mathcal{H})\longrightarrow\mathcal{A},\qquad\mathcal{A}=\Gamma^{q}\otimes I_{m}\subset\mathcal{B}_{h}(\mathcal{H}),
\]
fixed at $\theta_{0}$, such that
\begin{align}
\Tr\bigl[\rho_{\theta_{0}}\alpha(B)\bigr] & =\Tr(\rho_{\theta_{0}}B),\label{eq:sufficient1}\\
\Tr\bigl[(\partial_{i}\rho_{\theta_{0}})\alpha(B)\bigr] & =\Tr\bigl[(\partial_{i}\rho_{\theta_{0}})B\bigr],\qquad i=1,\ldots,d,\label{eq:sufficient2}
\end{align}
for every Hermitian operator $B$ on $\mathcal{H}$. Here $\Gamma^{q}$
has the representation specified in Eqs. (\ref{eq:def_spin_fac})
and (\ref{eq:def_spin_fac0}), and $I_{m}$ is the identity on a multiplicity
space. For every POVM $M$, the effects $\alpha(M_{x})$ form a POVM
with the same probabilities and first derivatives at $\theta_{0}$,
and hence the same Fisher information.

Every $\mathcal{A}$-valued POVM satisfies the same norm bound (\ref{eq:POVM_ineq}),
since tensoring with $I_{m}$ leaves operator norms unchanged. The
Bessel-inequality argument in Section \ref{sec:proof_main}, applied
with the SLDs and Fisher information of the original model, therefore
gives the same normalized trace bound. Spectral measurements of the
original normalized SLD directions realize the rank-one Fisher matrices
$vv^{\mathsf{T}}$ in normalized parameter coordinates, with $\|v\|=1$,
and randomization yields all of $\mathcal{K}_{d}$. Thus both the
exact Fisher-information region and the attainable weighted covariance
bound extend to this locally sufficient setting. The randomized spectral
measurements are constructed on the original Hilbert space; their
images under $\alpha$ are locally equivalent POVMs in $\mathcal{A}$.

A natural next question is how the attainable information region depends
on more general locally sufficient matrix Jordan algebras. The decomposition
recalled in Appendix \ref{sec:real_jordan} suggests studying the
other simple components and the combination of information within
and between direct-sum components. The spin-factor case provides an
exactly solvable starting point for this analysis. It links the Jordan
structure of a statistical model to its attainable measurement precision
and supplies explicit optimal measurements for models beyond qubits.

\section*{Author Contributions}

The author conceived the study, derived the theoretical results, performed
the numerical simulations, and wrote the manuscript. ChatGPT (OpenAI)
was used for English-language editing. The author reviewed all AI-assisted
text and takes full responsibility for the content. 

\section*{ACKNOWLEDGMENTS}

This work was supported by JSPS KAKENHI Grant Numbers JP23H01090 and
JP22K03466 and JST ERATO Grant Number JPMJER2402, Japan.

\appendix

\section{Matrix Jordan Algebras\label{sec:real_jordan}}

We recall the finite-dimensional structure underlying the discussion
in Sec. \ref{sec:conclusion}. Let $\mathcal{A}\subset\mathcal{B}_{h}(\mathcal{H})$
be a unital real Jordan subalgebra, with product $A\circ B=(AB+BA)/2$.
The Hilbert--Schmidt inner product is positive definite on $\mathcal{A}$
and satisfies
\[
\langle A\circ B,C\rangle_{\mathrm{HS}}=\langle A,B\circ C\rangle_{\mathrm{HS}}.
\]
Thus $\mathcal{A}$ is a Euclidean Jordan algebra. In particular,
it is formally real: if $\sum_{j}A^{2}_{j}=0$ for $A_{j}\in\mathcal{A}$,
then taking the trace gives $\sum_{j}\Tr(A^{2}_{j})=0$, so every
$A_{j}$ vanishes.

The Jordan--von Neumann--Wigner classification implies that $\mathcal{A}$
decomposes, as a real Jordan algebra, into simple ideals \cite{jordan1,jordan3,jordan2}:
\begin{equation}
\mathcal{A}\cong\bigoplus^{s}_{\alpha=1}\mathcal{J}_{\alpha}.\label{eq:jordan_deco}
\end{equation}
For a matrix Jordan algebra, the simple factors can be chosen from
the following nonredundant list:
\begin{equation}
\begin{gathered}\mathbb{R},\\
\mathbb{M}_{n}(\mathbb{R})_{h},\quad\mathbb{M}_{n}(\mathbb{C})_{h},\quad\mathbb{M}_{n}(\mathbb{H})_{h},\qquad n\geq3,\\
\Gamma^{q},\qquad q\geq2.
\end{gathered}
\label{eq:simple_matrix_jordan_factors}
\end{equation}
Here $\mathbb{H}$ denotes the quaternions, and $\mathbb{M}_{n}(\mathbb{F})_{h}$
is the real Jordan algebra of self-adjoint $n\times n$ matrices over
$\mathbb{F}$; for $\mathbb{F}=\mathbb{R}$, these are real symmetric
matrices. The exceptional Albert algebra of $3\times3$ Hermitian
octonionic matrices, which occurs in the abstract Euclidean classification,
is excluded because it has no faithful Jordan representation by complex
Hermitian matrices.

In the notation of this paper, $\Gamma^{q}$ has real dimension $q+1$.
The low-dimensional identifications are
\begin{equation}
\begin{aligned}\mathbb{M}_{2}(\mathbb{R})_{h} & \cong\Gamma^{2},\\
\mathbb{M}_{2}(\mathbb{C})_{h} & \cong\Gamma^{3},\\
\mathbb{M}_{2}(\mathbb{H})_{h} & \cong\Gamma^{5}.
\end{aligned}
\label{eq:low_dimensional_jordan_identifications}
\end{equation}
Also, $\mathbb{M}_{1}(\mathbb{F})_{h}\cong\mathbb{R}$ for $\mathbb{F}=\mathbb{R},\mathbb{C},\mathbb{H}$,
whereas $\Gamma^{1}\cong\mathbb{R}\oplus\mathbb{R}$ is not simple.
These identities explain the index ranges in Eq. (\ref{eq:simple_matrix_jordan_factors}).
The quaternionic matrix algebra has a standard faithful representation
by $2n\times2n$ complex matrices, whose restriction to the self-adjoint
part preserves the Jordan product. In particular, the $\Gamma^{5}$
model in Sec. \ref{sec:simulation} realizes the quaternionic $2\times2$
Hermitian Jordan algebra inside $\mathcal{B}_{h}(\mathbb{C}^{4})$. 

The abstract decomposition (\ref{eq:jordan_deco}) should be distinguished
from the choice of a concrete matrix representation \cite{jordan3,barnum_graydon_wilce}.
The units of the simple ideals correspond to mutually orthogonal projections
$e_{\alpha}\in\mathcal{A}$ with $\sum_{\alpha}e_{\alpha}=I$. They
give a Hilbert-space decomposition $\mathcal{H}=\bigoplus_{\alpha}\mathcal{H}_{\alpha}$,
where $\mathcal{H}_{\alpha}=e_{\alpha}\mathcal{H}$. After a unitary
change of basis, the algebra takes the form
\begin{equation}
U\mathcal{A}U^{*}=\left\{ \bigoplus^{s}_{\alpha=1}\pi_{\alpha}(a_{\alpha})\,\middle|\,a_{\alpha}\in\mathcal{J}_{\alpha}\right\} ,\label{eq:jordan_represented_decomposition}
\end{equation}
where each $\pi_{\alpha}:\mathcal{J}_{\alpha}\to\mathcal{B}_{h}(\mathcal{H}_{\alpha})$
is a faithful unital Jordan representation. The representations $\pi_{\alpha}$
retain the embedding and multiplicity information that is not specified
by the abstract simple factors alone.

For a fixed matrix representation $\mathcal{J}$, repeated copies
of the same representation can be written as
\[
\mathcal{J}\otimes I_{m}:=\{A\otimes I_{m}\mid A\in\mathcal{J}\},
\]
where $I_{m}$ is the identity on the multiplicity space $\mathbb{C}^{m}$.
This is the meaning of $\Gamma^{q}\otimes I_{m}$ in Sec. \ref{sec:conclusion}.
For spin factors, the amplified generators $\gamma_{j}\otimes I_{m}$
satisfy the same defining relations, and
\[
\|A\otimes I_{m}\|=\|A\|,\qquad\Tr(A\otimes I_{m})=m\,\Tr A.
\]
Consequently, the operator-norm bound for POVMs in Lemma \ref{lem:spin_povm_norm}
is unchanged by this multiplicity. The positive-projection lemma,
Lemma \ref{lem:proj-J}, applies to every unital matrix Jordan algebra,
while the two-eigenvalue structure of a spin factor supplies the additional
ingredient used to obtain the exact information region in Theorem
\ref{thm:main}.

\end{document}